\documentclass[11pt,reqno]{article}

\usepackage{amsmath,amssymb,amsthm}
\usepackage{epsfig}
\usepackage{epstopdf}
\usepackage[english]{babel}
\usepackage{caption}
\usepackage{a4wide}

\newtheorem{theorem}{Theorem}[section]
\newtheorem{remark}{Remark}[section]
\newtheorem{proposition}{Proposition}[section]

\begin{document}	
	
\title{Global Stability of a Diffusive SEIR Epidemic Model\\ with Distributed Delay\thanks{This is a 
preprint whose final form is published by Elsevier in the book 
'Mathematical Analysis of Infectious Diseases', 1st Edition -- June 1, 2022, 
ISBN: 9780323905046.}}

\author{Abdesslem Lamrani Alaoui$^1$\\ 
\texttt{abdesslemalaoui@gmail.com}
\and Moulay Rchid Sidi Ammi$^2$\\ 
\texttt{rachidsidiammi@yahoo.fr}
\and Mouhcine Tilioua$^1$\\ 
\texttt{m.tilioua@umi.ac.ma}
\and Delfim F. M. Torres$^{3,}$\thanks{Corresponding author.}\\ 
\texttt{delfim@ua.pt}}
	
\date{$^1$MAMCS Group, MAIS Laboratory, FST Errachidia,\\
Moulay Isma\"{\i}l University of Mekn\`{e}s, Morocco\\[0.3cm]
$^2$AMNEA Group, MAIS Laboratory, FST Errachidia,\\
Moulay Isma\"{\i}l University of Mekn\`{e}s, Morocco\\[0.3cm]
$^3$Center for Research and Development in Mathematics and Applications (CIDMA),
Department of Mathematics, University of Aveiro, 3810-193 Aveiro, Portugal}

\maketitle


\begin{abstract}
We study the global dynamics of a reaction-diffusion SEIR infection model 
with distributed delay and nonlinear incidence rate. The well-posedness 
of the proposed model is proved. By means of Lyapunov functionals, we show 
that the disease free equilibrium state is globally asymptotically stable 
when the basic reproduction number is less or equal than one, 
and that the disease endemic equilibrium is globally asymptotically stable 
when the basic reproduction number is greater than one. Numerical simulations 
are provided to illustrate the obtained theoretical results.

\medskip


\medskip

\end{abstract}	
	

\section{Introduction}
	
It is well-known that mathematical models help the understanding  
of disease dynamics, giving suggestions for the control of the 
spread of diseases in a population, both in time and space.
It turns our that the spatial spread of many human diseases 
is affected by how, where, and when people are moving. For instance, 
it has been proved that human movement has played a key role 
in the dynamics of influenza \cite{SAMSUZZOHA20103461,SAMSUZZOHA20115507} 
and malaria \cite{Bai2018}. Movement affects pathogen dynamics in two main ways: 
it may introduce pathogens into susceptible populations 
or it may increase the contact between susceptible and infected individuals. 
This means that individuals are also affected by diseases transmission
on the basis of social, demographic, and geographic factors.

In the literature, there are some mathematical studies that investigate the influence 
of the spatial aspect of host populations on the dynamics of diseases  
\cite{Banerjee2016,Hwang2013,Lou2011,Wang2018,Xu2017,Nauman2019,Jinliang2017,KIM20131992}. 
However, many authors still propose models in which it is assumed
the environment to be uniformly mixed, without taking into account 
the location or mobility of the populations 
\cite{CAPASSO197843,Abdesslem,enatsu2012lyapunov,MR2762609,MR2130156,%
MR2879132,WANG20102390,gumel2003qualitative,LI2014100,Ruan2003135,Elazzouzi}. 
Thus, it is appropriate to investigate the spacial dimension into many 
of the available models.  

In recent years, the spatial transmission dynamics of delayed models 
has attracted the attention of many researchers \cite{Jing2011,Xia2018}. 
In \cite{CONNELLMCCLUSKEY201564}, McCluskey and Yang propose a model of virus 
dynamics that includes diffusion and time delay. They show that the 
equilibria of the system are globally asymptotically stable. 
In \cite{Jinliang2017}, Kuniya and Wang consider a spatially 
diffusive SIR epidemic model and discuss the global stability 
analysis of equilibria for two special cases: the case of no diffusive 
susceptible individuals and that of no diffusive infective individuals. 
Xu and Chen study the dynamics of an SIS epidemic model with diffusion 
\cite{Xu2017}. First, they establish the well-posedness of the model. 
Then, by using the linearization method and constructing a suitable 
Lyapunov function, they show the local and global stability of 
the disease-free equilibrium and of the endemic equilibrium, respectively. 
In \cite{junjie2020}, Yang and Wei investigate  a delayed reaction-diffusion 
virus model with a general incidence function and spatially dependent parameters. 
They derive the basic reproduction number for the model 
and prove the uniform persistence of solutions 
and the global interactivity of the equilibria. 

Motivated by the discussions above and the work \cite{McCluskey2010837}
of McCluskey, here we focus ourselves on the global stability analysis of 
a general SEIR epidemiological model with diffusion and distributed delay. 
In some sense, the present work can be viewed as a continuation 
and generalization of \cite{McCluskey2010837},
where an SIR disease model is investigated.
In contrast, here we study a generalized SEIR epidemic model 
with distributed delay and a nonlinear incidence function.  
Moreover, it is necessary to point out that the delay in our model 
represents the incubation time taken to become infectious.  
Our goal is to investigate the impact of the spatial dimension 
on the dynamic behavior of the considered model. Furthermore, 
we discuss the global stability of the model near equilibria 
(the disease-free equilibrium $E_0$ and the disease-endemic 
equilibrium $E^*$) by means of Lyapunov's method. Finally, 
to illustrate the obtained theoretical results, 
some numerical simulations are carried out.

The text is organized as follows. The mathematical model to be studied
is formulated in Section~\ref{sec:2}. In Section~\ref{sec:3}, we provide 
a mathematical analysis of the considered model. More precisely, 
we show that the model is well-posed, we compute the basic reproduction number $R_0$
and the equilibria, proving their global stability. In Section~\ref{sec:4},
a numerical example, with an incidence function satisfying the assumptions considered, 
is given and discussed. We finish with Section~\ref{sec:5}, 
providing some concluding remarks.


\section{Mathematical model}
\label{sec:2}

We are interested in a general SEIR epidemic model 
with distributed delay and diffusion. The dynamics 
is governed by the following system of equations:
\begin{equation}
\label{Model}
\left\{
\begin{array}{lll}
\dfrac{\partial S(x,t)}{\partial t} - k_S \Delta S(x,t)
= b- \mu S(x,t)- \beta \displaystyle \int_0^{h}
g(\tau)f(S(x,t),I(x,t-\tau))d\tau, \quad x\in\Omega,\\[2ex]
\dfrac{\partial E(x,t)}{\partial t} - k_E \Delta E(x,t)
= \beta \displaystyle \int_0^{h}g(\tau)f(S(x,t),I(x,t-\tau))d\tau 
-(\mu+\alpha) E(x,t), \quad x\in\Omega,\\[2ex]
\dfrac{\partial I(x,t)}{\partial t} - k_I \Delta I(x,t)
= \alpha E(x,t) -(\mu+c+\gamma)I(x,t), \quad x\in\Omega,\\[2ex]
\dfrac{\partial R(x,t)}{\partial t}- k_R \Delta R(x,t)
=\gamma I(x,t)- \mu R(x,t), \quad x\in\Omega,\\[2ex]
\displaystyle \frac{\partial S(x,t)}{\partial \nu}
=\frac{\partial E(x,t)}{\partial \nu}
=\frac{\partial I(x,t)}{\partial \nu}
=\frac{\partial R(x,t)}{\partial \nu}=0,
\quad x\in\partial\Omega,
\end{array}
\right.
\end{equation}
where $t>0$; $\Omega$ is a bounded domain in $\mathbb{R}^n$ 
with smooth boundary $\partial\Omega$; $\nu$ is the outward 
normal to $\partial\Omega$; $k_S>0$, $k_E>0$, $k_I>0$ and
$k_R>0$ stand for the diffusion rates; $S(x,t)$, $E(x,t)$, 
$I(x,t)$ and $R(x,t)$ denote the number of susceptible, 
exposed, infected and recovered individuals at time $t$ 
in position $x$, respectively; $b$ is the recruitment 
rate of the population; $\mu$ is the natural death rate 
of the population; $\gamma$ represents the natural recovery 
rate of infective individuals; $c$ is the death rate of the 
population caused by the infection; and $\beta$ represents the 
transmission coefficient. Individuals leave the susceptible class 
at a rate 
$$ 
\int_0^{h}g(\tau)f(S(x,t),I(x,t-\tau))d\tau,  
$$
where $h$ represents the maximum time taken to become infectious and 
$g$ is a non-negative function satisfying 
$\displaystyle \int_0^{h}g(\tau)d\tau = 1$.
	
The initial condition for the above system 
is given for $\theta \in[-h,0] $ by
\begin{equation*}
\begin{split}
\Phi(\theta)(x)
&= \left(\Phi_1 (x,\theta ), \Phi_2 (x,\theta ), 
\Phi_3 (x,\theta ), , \Phi_4 (x,\theta )\right)\\
&=\left(S(x,\theta), E(x,\theta), I(x,\theta), R(x,\theta)\right),
\quad  x \in\overline{\Omega},
\end{split}
\end{equation*}
with $\Phi \in~ C([-h,0], \mathbb{X})$.  Here, $C([-h,0], \mathbb{X})$ 
denotes the space of continuous functions mapping from $ [-h,0] $ 
to $\mathbb{X}$ equipped with the sup-norm and 
$\mathbb{X}=C(\Omega, \mathbb{R}^{+4})$ denotes the space 
of continuous functions mapping from $\Omega $ to $\mathbb{R}^{+4}$.
	
Our main objective is to discuss the global stability of the SEIR model \eqref{Model}. 
For that, we will construct suitable Lyapunov functions.
	
Throughout this work, we assume that $f:\mathbb{R}^2_+ \rightarrow \mathbb{R}^{+}$ 
is continuously differentiable in the interior of $ \mathbb{R}^{+} $ with
$$
f(0, I) = f(S, 0) = 0 \mbox{ \ for\ } S, I\geqslant 0
$$  
and the following hypotheses hold:
\begin{enumerate}
\item[$(H_1)$] $f(S,I)$ is a strictly monotone increasing 
function of $S \geqslant 0$ for any fixed $I > 0$
and a monotone increasing function of $I > 0$ 
for any fixed $S \geqslant 0$;

\item[$(H_2)$] $\phi(S, I) = \dfrac{f(S,I )}{I}$ 
is a bounded and monotone decreasing function of $I > 0$ 
for any fixed $ S \geqslant 0 $ and 
$k(S) = \lim\limits_{I \rightarrow 0^+} \phi(S, I)$ 
is a continuous and monotone increasing function on $S \geqslant 0$. 
\end{enumerate}		


\section{Analysis of the model}
\label{sec:3}

In this section, we show that our model \eqref{Model} is well-posed 
(Section~\ref{sec:3.1}), we compute its equilibria and its 
basic reproduction number $R_0$ (Section~\ref{sec:3.2})
and prove the global stability of the disease free 
(Section~\ref{sec:3.3}) and endemic (Section~\ref{sec:3.4})
equilibrium points. 


\subsection{Well-posedness}
\label{sec:3.1}

Let $A$ be the operator defined on $\mathbb{X}$ as follows:
\begin{equation}
\label{4.3}
\begin{tabular}{llll}
$A:$ &$D(A) \subset \mathbb{X}$& $\longrightarrow$ & $\mathbb{X}$\\
& $u$ & $\longmapsto$ &$Au(x)=(K_{S} \Delta u_1, K_{E} \Delta u_2, 
K_{I} \Delta u_3, K_{R} \Delta u_4)$,
\end{tabular}
\end{equation}
where 
$$ 
D(A):=\left\{u\in \mathbb{X}: \Delta u \in \mathbb{X}, 
\dfrac{\partial u}{\partial \nu}=0 \ \mbox{on} \ \partial \Omega \right\}.
$$	
Then, $A$ is the infinitesimal generator of a strongly continuous 
semi-group  $ \exp(tA)$ in $\mathbb{X}$. For any function  
$u : [-h, \sigma ) \longrightarrow \mathbb{X}$ with some  $\sigma> 0$, 
we define $u_t \in C([-h,0], \mathbb{X})$ by 
$u_t(\theta ) = u(t + \theta )$, $\theta \in [-h, 0]$.

Let $F$ be a function defined by
$$
\begin{array}{cccc}
F : &C([-h, 0], \mathbb{X} )&\longrightarrow &\mathbb{X}\\
&\phi& \longmapsto &F(\phi),
\end{array}
$$
where
\begin{equation*}
F(\phi)=F(\phi_1, \phi_2, \phi_3, \phi_4) 
= \left(
\begin{array}{ccc}
b- \mu \phi_1(x,0)- \beta 
\displaystyle \int_0^{h}g(\tau)f(\phi_1(x,0),\phi_3(x,-\tau))d\tau\\[2ex]
\beta \displaystyle \int_0^{h}g(\tau)f(\phi_1(x,0),\phi_3(x,-\tau))d\tau 
-(\mu+\alpha) \phi_2(x,0)\\[2ex]
\alpha \phi_2(x,0) -(\mu+c+\gamma)\phi_3(x,0)\\[2ex]
\gamma \phi_3(x,0)- \mu \phi_4(x,0)
\end{array}
\right).
\end{equation*} 
Function $F$ is locally Lipschitzian on $C([-h,0], \mathbb{X})$. 
In addition, the system \eqref{Model} can be written 
in the following abstract form:
\begin{equation}
\label{4.5}
\left\{
\begin{array}{lll}
\dfrac{du(t)}{dt}=Au(t) + F(u_t), \quad t>0, \\[2ex]
u_0= \Phi,
\end{array}
\right.
\end{equation}
where $u(t) = \left(S(\cdot, t), E(\cdot, t), I(\cdot, t), R(\cdot, t)\right)^T$ 
and $\Phi = \left(S(\cdot, 0), E(\cdot, 0), I(\cdot, 0), R(\cdot, 0)\right)^T$. 	

According to the well-known theory of differential delay equations, see e.g. 
\cite{hale,hale2}, the following existence and uniqueness result holds.

\begin{proposition}[See \cite{hale,hale2}]
\label{prop:ex:uniq}
Let $A$ be defined by \eqref{4.3}. For each $u_0 \in  \mathbb{X}$ 
there exists a unique solution  $u: [0, T_{max}] \rightarrow \mathbb{X}$ 
of system \eqref{4.5} on the maximal interval $[0, T_{max}]$ such that
$$
u(t) = T(d t)\phi(0) + \int_{0}^{t} T(d (t-s))F(u_s)ds, 
\quad t\geq 0,
$$
where either $T_{max}= +\infty$ or 
$\lim\limits_{t\rightarrow T_0} \sup \lVert u(t)\lVert_\mathbb{X} = +\infty$. 
\end{proposition}

We now prove the boundedness of the solution.

\begin{proposition}
\label{prop:bound}
If 
\begin{equation}
\label{eq:sol:Modelo}
\left(S(\cdot, t), E(\cdot,t), I(\cdot, t), R(\cdot, t)\right)
\end{equation}
is the solution of \eqref{Model}, then \eqref{eq:sol:Modelo} 
is bounded.
\end{proposition}

\begin{proof}
Adding the three equations of system \eqref{Model}, we obtain that
\begin{multline*}
\dfrac{\partial S(x,t)}{\partial t}
+ \dfrac{\partial E(x,t)}{\partial t}
+\dfrac{\partial I(x,t)}{\partial t}
+\dfrac{\partial R(x,t)}{\partial t} 
- k_S \Delta S(x,t) - k_E \Delta E(x,t)\\
- k_I \Delta I(x,t) -k_R \Delta R(x,t)
=b- \mu S(x,t) -(\mu+c)I(x,t) )- \mu R(x,t).
\end{multline*}
Integrating both sides,
\begin{equation*}
\begin{split}
\displaystyle \int_{\Omega}^{} &
\Big \{ \dfrac{\partial S(x,t)}{\partial t}+ \dfrac{\partial E(x,t)}{\partial t} 
+\dfrac{\partial I(x,t)}{\partial t}+\dfrac{\partial R(x,t)}{\partial t} \Big \}dx\\
&\quad - \displaystyle \int_{\Omega}^{} \Big \{ k_S \Delta S(x,t)
+  k_E \Delta E(x,t) + k_I \Delta I(x,t) +k_R \Delta R(x,t)\Big \}dx\\
&= \displaystyle \int_{\Omega}^{} 
\Big\{b- \mu S(x,t) -(\mu+c)I(x,t) )- \mu R(x,t)\Big \}dx. 
\end{split}
\end{equation*}
By Green's formula, we obtain that
\begin{equation*}
\begin{split}
k_S \int_{\Omega}^{} \  \Delta S(x,t) dx 
&= k_S \displaystyle \int_{\partial\Omega}^{}  
\dfrac{\partial S(x,t)}{\partial \nu} dx,\\ 
k_E \int_{\Omega}^{} \  \Delta E(x,t) dx 
&= k_E \displaystyle \int_{\partial\Omega}^{}  
\dfrac{\partial E(x,t)}{\partial \nu} dx,\\
k_I \int_{\Omega}^{} \  \Delta I(x,t) dx 
&= k_I \displaystyle \int_{\partial\Omega}^{}  
\dfrac{\partial I(x,t)}{\partial \nu} dx,\\ 
k_R \int_{\Omega}^{} \  \Delta R(x,t) dx 
&= k_R \displaystyle \int_{\partial\Omega}^{}  
\dfrac{\partial R(x,t)}{\partial \nu} dx.
\end{split} 
\end{equation*}		
From the Neumann boundary conditions, we have that
$$
\dfrac{\partial S(x,t)}{\partial \nu}
=\dfrac{\partial E(x,t)}{\partial \nu}
=\dfrac{\partial I(x,t)}{\partial \nu}
=\dfrac{\partial R(x,t)}{\partial \nu}=0, 
\quad x\in \partial\Omega,\quad t>0. 
$$	
Hence,
\begin{equation*}
\begin{split}
\displaystyle \int_{\Omega}^{} & 
\left\{ \dfrac{\partial S(x,t)}{\partial t} 
+ \dfrac{\partial E(x,t)}{\partial t}
+\dfrac{\partial I(x,t)}{\partial t}+\dfrac{\partial R(x,t)}{\partial t}\right\}dx\\
&=\displaystyle \int_{\Omega}^{} \Big \{b
- \mu \Big(S(x,t)+E(x,t)+I(x,t)+R(x,t)\Big) -cI(x,t)\Big \}dx  \\[2ex]
&\leq \displaystyle \int_{\Omega}^{} \Big \{b- \mu \Big(S(x,t)
+E(x,t)+I(x,t)+R(x,t)\Big)  \Big\}dx \\[2ex]
&= b \arrowvert \Omega \arrowvert
-\displaystyle \int_{\Omega}^{} 
\Big\{ \mu \Big(S(x,t)+E(x,t)+I(x,t)+R(x,t)\Big)\Big\}dx. 
\end{split}
\end{equation*}		
Denote 	
$$
\int_{\Omega}^{} \Big \{  \Big(S(x,t)
+E(x,t)+I(x,t)+R(x,t)\Big)  \Big \}dx= N(t),
$$		
which gives
$$ 
\dfrac{dN(t)}{dt} \leq   b \arrowvert \Omega \arrowvert- \mu N(t).  
$$	
It follows that
$$ 
0 \leq N(t) \leq \dfrac{b \arrowvert \Omega \arrowvert}{\mu}
+ N(0) \exp(-\mu t).
$$	
Therefore, 
$$
N(t) \leq \max \Big\{\dfrac{b 
\arrowvert \Omega \arrowvert}{\mu}, N(0) \Big\},
$$
where
$$
\begin{array}{lll}
N(0)&=&\displaystyle \int_{\Omega}^{} 
\Big\{ \Big(S(x,0)+E(x,0)+I(x,0)+R(x,0)\Big)  \Big \}dx \\[2ex]
& \leq & \displaystyle \int_{\Omega}^{}  
\|S(x,0)+I(x,0)+R(x,0)\|_\infty  dx \\[2ex]
&=& \|S(x,0)+E(x,0)+I(x,0)+R(x,0)\|_\infty \arrowvert \Omega \arrowvert. 
\end{array}
$$	
This shows that 
$N(t)=\displaystyle \int_{\Omega}^{}\Big\{\Big(S(x,t)+E(x,t)+I(x,t)+R(x,t)\Big)\Big\}dx$ 
is bounded.
\end{proof}

\begin{remark}
The local existence and uniqueness (Proposition~\ref{prop:ex:uniq})
and the boundedness of the solution (Proposition~\ref{prop:bound})
of (\ref{Model}) implies the global existence and uniqueness
of the solution.
\end{remark}

Since the  first two equations in (\ref{Model}) 
do not contain $R(x,t)$, it is sufficient to analyze 
the behavior of solutions to the following system:
\begin{equation}
\label{MM}
\left\{
\begin{array}{lll}
\dfrac{\partial S(x,t)}{\partial t} - k_S \Delta S(x,t)
= b- \mu S(x,t)- \beta \displaystyle \int_0^{h}g(\tau)f(S(x,t),I(x,t-\tau))d\tau, 
\quad x\in\Omega,\\[2ex]
\dfrac{\partial E(x,t)}{\partial t} - k_E \Delta E(x,t)
= \beta \displaystyle \int_0^{h}g(\tau)f(S(x,t),I(x,t-\tau))d\tau 
-(\mu+\alpha) E(x,t), \quad x\in\Omega,\\[2ex]
\dfrac{\partial I(x,t)}{\partial t} - k_I \Delta I(x,t)
= \alpha E(x,t) -(\mu+c+\gamma)I(x,t), \quad x\in\Omega,\\[2ex]
\displaystyle \frac{\partial S(x,t)}{\partial \nu}
=\frac{\partial E(x,t)}{\partial \nu}
=\frac{\partial I(x,t)}{\partial \nu}=0,
\quad x\in\partial\Omega,
\end{array}
\right.
\end{equation}
$t>0$. In the sequel we use this fact.


\subsection{Equilibria and the basic reproduction number}
\label{sec:3.2}

System (\ref{MM}) always has a disease-free equilibrium 
$E_{0} = (S_{0},0, 0)$, where $ S_{0} =~\dfrac{b}{\mu}$. 
Furthermore, by a simple and direct calculation, we conclude 
that the basic reproduction number for the model is given by
$$ 
R_{0}=\dfrac{\beta \alpha 
\dfrac{\partial f(E_0)}{\partial I}}{(\mu+\alpha)(\mu+\gamma+c)}.
$$
We have the following result.

\begin{theorem}
If  $R_0 > 1$, then (\ref{MM}) admits a unique endemic equilibrium 
$E^* = (S^*, E^*, I^*)$.
\end{theorem}

\begin{proof}
We look for solutions $(S^*, E^*, I^*)$ of the equations  
$\dfrac{\partial S}{\partial t} = 0$  and 
$\dfrac{\partial E}{\partial t} = 0$. First note that  
$\dfrac{\partial S}{\partial t} + \dfrac{\partial E}{\partial t} = 0$  
implies
$$ 
b- \mu S^* -(\mu+\alpha)E^*  =0 
$$
and so 
$$ 
S^*=~\dfrac{b}{\mu}-\dfrac{(\mu+\alpha)(\mu+c+\gamma)I^*}{\mu\alpha}. 
$$
Let $H$ be a function defined for $\mathbb{R}^+ $ to $\mathbb{R}$ by
$$ 
H(I)=\beta\dfrac{f \Big(\dfrac{b}{\mu}
-\dfrac{(\mu+\alpha)(\mu+c+\gamma)I}{\mu\alpha},I \Big)}{I}
-\dfrac{(\mu+\alpha)(\mu+c+\gamma)}{\alpha}.
$$
By the hypotheses $(H_1)$ and $(H_2)$, 
$H$ is strictly monotone decreasing on $\mathbb{R}^+$ satisfying
\begin{equation*}
\begin{split}
\lim\limits_{I \rightarrow 0^+} H(I) 
&= \beta\dfrac{\partial f(E_{0})}{\partial I}
-\dfrac{(\mu+\alpha)(\mu+c+\gamma)}{\alpha}\\  
&= \dfrac{(\mu+\alpha)(\mu+c+\gamma)}{\alpha}(R_0-1)\\
&> 0
\end{split}
\end{equation*}
and 
$$ 
H \left(\dfrac{b\alpha}{(\mu+\alpha)(\mu+c+\gamma)} \right)
=-\dfrac{(\mu+\alpha)(\mu+c+\gamma)}{\alpha} <0, 
$$
which implies that there exists a unique positive solution $I=I^*$ 
such that
$$ 
0 < I^* <\dfrac{b\alpha}{(\mu+\alpha)(\mu+c+\gamma)}.
$$
The proof is complete.
\end{proof}

In what follows we study the stability of $E_0$ and $E^*$.


\subsection{Global stability of the disease free equilibrium}
\label{sec:3.3}

In this section, we show the global asymptotic stability 
of the disease-free equilibrium $E_0$  of the system (\ref{MM}) 
by  constructing a Lyapunov functional. The following result holds.

\begin{theorem} 
\label{theorem1}
Under hypotheses $(H_1)$ and $(H_2)$ 
the disease free equilibrium $E_{0} $ of system \eqref{MM} 
is globally asymptotically stable if, and only if, $R_{0} \leq 1$.
\end{theorem}

\begin{proof} 
\label{fun1}
To prove our result, we consider the following Lyapunov functional: 	
$$ 
V(t)= \int_{\Omega}^{} \Big \{V_1(t)+V_2(t)\Big\}dx, 
$$
where
$$ 
V_1(t)= \int_{S_0}^{S(x,t)} 
\Big(1-\dfrac{k( S_0)}{k( \sigma)} \Big)d \sigma
+ E(x,t)+\dfrac{\mu + \alpha}{\alpha}I(x,t)  
$$ 
and 
$$ 
V_2(t)=\dfrac{\mu + \alpha}{\alpha}(\mu+c+\gamma) 
\int_0^h g( \tau) \int_{t-\tau}^{t} I(u)dud\tau. 
$$
Then,
\begin{equation*}
\begin{split}
\dfrac{dV(t)}{dt}
&=\displaystyle \int_{\Omega}^{} \Big\{\Big(1-\dfrac{k( S_0)}{k(S(x,t))}\Big)\\
&\qquad \times \Big(k_S \Delta S(x,t)+ b- \mu S(x,t)- \beta \displaystyle 
\int_{0}^{h}g(\tau)f(S(x,t),I(x,t-\tau))d\tau \Big)\\
&\quad + k_E \Delta E(x,t)+ \beta \displaystyle 
\int_0^{h}g(\tau)f(S(x,t),I(x,t-\tau))d\tau -(\mu+\alpha) E(x,t)\\
&\quad+ \dfrac{\mu + \alpha}{\alpha}\Big\{ k_I \Delta I(x,t)
+ \beta \displaystyle \int_0^h g(\tau)f(S(x,t),I(x,t-\tau))d\tau -(\mu+c+\gamma)I(x,t)\Big\}\\
&\quad+ \dfrac{\mu + \alpha}{\alpha}(\mu+c+\gamma) 
\displaystyle \int_{0}^{h} g( \tau) ( I(x,t) - I(x,t-\tau))  d\tau \Big \}dx \\
&= \displaystyle \int_{\Omega}^{} \Big\{ -\mu \Big(1-\dfrac{k( S_0)}{k(S(x,t))} \Big) 
\Big(S(x,t)- S_0 \Big) \\
&\quad-\Big(1-\dfrac{k( S_0)}{k(S(x,t))} \Big) \Big(\beta \displaystyle 
\int _{0}^{h}g(\tau)f(S(x,t),I(x,t-\tau))d\tau \Big)\\ 
&\quad+ \beta \displaystyle \int _{0}^{h}g(\tau)f(S(x,t),I(x,t-\tau))d\tau  
- \dfrac{\mu + \alpha}{\alpha}(\mu+c+\gamma)I(x,t) \\
&\quad+ \dfrac{\mu + \alpha}{\alpha}(\mu+c+\gamma) \displaystyle 
\int_{0}^{h} g( \tau) ( I(x,t) - I(x,t-\tau))  d\tau \Big \}dx \\
&\quad+ \displaystyle \int_{\Omega}^{} k_S\Delta S(x,t)dx-k_S \displaystyle 
\int_{\Omega}^{}\dfrac{k( S_0)}{k(S(x,t))} \Delta S(x,t)dx \\
&\quad + \displaystyle \int_{\Omega}^{} k_E \Delta E(x,t)dx
+ \displaystyle \int_{\Omega}^{} k_I \Delta I(x,t)dx
\end{split}
\end{equation*}
from which we conclude that
\begin{equation*}
\begin{split}
\dfrac{dV(t)}{dt}
&= \displaystyle \int_{\Omega}^{} \Big\{ -\mu 
\Big(1-\dfrac{k( S_0)}{k(S(x,t))} \Big) \Big(S(x,t)- S_0 \Big)\\
&\qquad+ \displaystyle \int_{0}^{h} g(\tau) \Big( 
\beta \alpha\dfrac{\phi(S(x,t),I(x,t - \tau))}{(\mu+\alpha)(\mu+c+\gamma)}
\dfrac{k( S_0)}{k(S(x,t))} -1  \Big)\\
&\qquad\quad \times \dfrac{\mu + \alpha}{\alpha}(\mu+c+\gamma)
I(x,t - \tau)d\tau\Big \}dx \\
&\quad+  \displaystyle \int_{\Omega}^{} k_S\Delta S(x,t)dx
-k_S \displaystyle \int_{\Omega}^{}\dfrac{k( S_0)}{k(S(x,t))} 
\Delta S(x,t)dx \\
&\quad + \displaystyle \int_{\Omega}^{} k_E \Delta E(x,t)dx
+ \displaystyle \int_{\Omega}^{} k_I \Delta I(x,t)dx.
\end{split}
\end{equation*}
By Green's formula and from the Neumann boundary conditions, we obtain that		
\begin{multline*}
\displaystyle  \int_{\Omega}^{} \left(k_S \Delta S(x,t) + k_E  \Delta E(x,t) 
+ k_I  \Delta I(x,t)\right) dx \\
=  \displaystyle \int_{\partial \Omega}^{} 
\left(k_S \dfrac{\partial S(x,t)}{\partial \nu} 
+ k_E \dfrac{\partial E(x,t)}{\partial \nu} 
+ k_I \dfrac{\partial I(x,t)}{\partial \nu}\right) dx=0
\end{multline*}
and
$$
k_{S} \displaystyle \int_{\Omega}^{}\dfrac{k(S_0)}{k(S(x,t))} \Delta S(x,t)dx
= k_{S} \dfrac{k(S_0)}{k(S(x,t))}\displaystyle \int_{\Omega}^{}
\dfrac{\partial k(S)}{\partial S} \dfrac{|\Delta S(x,t)|^2}{(k(S))^2}dx.
$$
It follows that
\begin{equation*}
\begin{split}
\dfrac{dV(t)}{dt}
&= \displaystyle \int_{\Omega}^{} \Big\{ -\mu 
\Big(1-\dfrac{k( S_0)}{k(S(x,t))} \Big) \Big(S(x,t)- S_0 \Big)\\
&+ \displaystyle \int_{0}^{h} g(\tau) \Big( \beta \alpha
\dfrac{\phi(S(x,t),I(x,t - \tau))}{(\mu+\alpha)(\mu+c+\gamma)}
\dfrac{k( S_0)}{k(S(x,t))} -1  \Big)\dfrac{\mu 
+ \alpha}{\alpha}(\mu+c+\gamma)I(x,t - \tau)d\tau\Big \}dx \\
&- k_{S} \dfrac{k( S_0)}{k(S(x,t))}\displaystyle 
\int_{\Omega}^{}\dfrac{\partial k(S)}{\partial S} 
\dfrac{|\Delta S(x,t)|^2}{(k(S))^2} dx. 
\end{split}
\end{equation*}
Since $k(S)$ is a monotone increasing function with respect to $S$, 
one has $\dfrac{\partial k(S)}{\partial S}\geq 0$. From hypothesis $(H_1)$, 
we get
$$
-\mu \Big(1-\dfrac{k( S_0)}{k(S(x,t))} \Big) \Big(S(x,t)- S_0 \Big) 
\leq 0
$$
and, from hypothesis $(H_2)$,
$$ 
\beta \alpha\dfrac{\phi(S(x,t),I(x,t - \tau))}{(\mu+\alpha)(\mu+c+\gamma)}
\dfrac{k( S_0)}{k(S(x,t))} \leq \beta \alpha\dfrac{k(S(x,t))}{(\mu
+\alpha)(\mu+c+\gamma)}\dfrac{k( S_0)}{k(S(x,t))} = R_0.
$$
If $R_0 \leq 1$, then
$$
\begin{array}{lll}
\dfrac{dV(t)}{dt} & \leq &\displaystyle \int_{\Omega}^{} \Big\{ 
-\mu \Big(1-\dfrac{k( S_0)}{k(S(x,t))} \Big) \Big(S(x,t)- S_0 \Big)
+ ( R_0 -1 )(\mu+c+\gamma)I(x,t - \tau ) \Big \}dx\\[2ex]
&-& k_{S} \dfrac{k( S_0)}{k(S(x,t))}\displaystyle \int_{\Omega}^{}
\dfrac{\partial k(S)}{\partial S} \dfrac{|\Delta S(x,t)|^2}{(k(S))^2}dx.
\end{array}
$$
Clearly, $\dfrac{dV(t)}{dt}\leq 0$ for all $t>0$ and $S, I, R>0$ and 
$\dfrac{dV(t)}{dt}= 0$ if, and only if, $(S, E, I)=(S_0, 0, 0)$, 
the largest compact invariant set in $\left \{(S, E, I): \dfrac{dV(t)}{dt}= 0\right \}$ 
being $E_0 $. By applying LaSalle's invariance principle \cite[Theorem~4.3.4]{Daniel1981},  
we conclude that the disease-free equilibrium point $E_0$ of system \eqref{MM} 
is globally asymptotically stable when $R_0\leq 1$, which completes the proof.
\end{proof} 
	

\subsection{Global stability of the endemic equilibrium}
\label{sec:3.4}

Now, we show the global asymptotic stability  
of the endemic equilibrium $E^*$ of system (\ref{MM}). 
As in the proof of Theorem~\ref{theorem1}, we construct 
a suitable Lyapunov functional and make our conclusion 
with the help of LaSalle's invariance principle.
	
\begin{theorem} 
\label{theorem2}
Assume that hypotheses $(H_1)$ and $(H_2)$ hold. 
If $R_{0} > 1$, then the endemic equilibrium of system \eqref{MM} 
is the only equilibrium and is globally asymptotically stable.
\end{theorem}

\begin{proof}
Let $G$ be the function defined from $\mathbb{R}^+ $ to $\mathbb{R}$  by
$$ 
G(x)=x-1-\ln(x).
$$ 
We have $G(x)\geq 0 $ if $x > 0$ and $G(x)=0 $ if $x=1$.
Let us consider the following Lyapunov functional: 
$$ 
W(t)= \int_{\Omega}^{} \Big (W_1(t)+W_2(t)\Big) dx,
$$
where
$$
W_1(t) = S(x,t)-S^* - \int_{S^*} ^{S(x,t)} \dfrac{f(S^*,I^*)}{f( \sigma ,I^*)}d \sigma 
+\dfrac{\mu+\alpha}{\alpha} \left(I(x,t)-I^* -I^* \ln\left(\dfrac{I(x,t)}{I^*}\right)\right)
$$
and
$$ 
W_2(t) =  \left(E(x,t)-E^* -E^* \ln\left(\dfrac{E(x,t)}{E^*}\right)\right).  
$$
Then,
\begin{equation*}
\begin{split}
\dfrac{dW(t)}{dt}
&= \left(1-\dfrac{f(S^*,I^*)}{f(S(x,t),I^*)} \right) 
\Big(( k_{S} \Delta S(x,t)+b- \mu S(x,t)\\
&- \beta \displaystyle \int_0^{h} g(\tau)f(S(x,t),I(x,t- \tau))d\tau\Big)\\
&+\dfrac{\mu+\alpha}{\alpha} \left(1-\dfrac{I^*}{I(x,t)}\right) 
\Big( k_{I} \Delta I(x,t)+\alpha E(x,t)-(\mu+c+ \gamma)I(x,t)\Big)\\
&+ \left(1-\dfrac{E^*}{E(x,t)}\right) \Big( k_{E} \Delta E(x,t)
+\beta \displaystyle \int_0^{h}g(\tau)f(S(x,t),I(x,t-\tau))d\tau 
-(\mu+\alpha)E(x,t) \Big)
\end{split}
\end{equation*}
and
$$
\left\{
\begin{array}{ll}
b= \mu S^* + \beta f(S^*,I^*),\\[2ex]
\beta f(S^*,I^*)=(\mu+\alpha) E^* ,\\[2ex]
\alpha E^*= ( \mu+c+ \gamma)I^*. 
\end{array}
\right.
$$
It follows that
\begin{equation*}
\begin{split}
\dfrac{dW(t)}{dt}
&= \displaystyle \int_{\Omega}^{}  \mu \left(1
-\dfrac{f(S^*,I^*)}{f(S(x,t),I^*)}\right) \Big(S(x,t)-S^* \Big) dx
+ \beta f(S^*,I^*)\\
&\times \displaystyle \int_{\Omega}^{}  
\displaystyle \int_0^{h}g(\tau) \left(3-\dfrac{f(S^*,I^*)}{f(S(x,t),I^*)} 
- \dfrac{I(x,t)}{I^*} -\dfrac{EI^*}{I E^*} 
-\dfrac{f(S(x,t),I(x,t- \tau))E^*}{f(S^*,I^*)E}\right) dx\\[2ex]
&+ \displaystyle \int_{\Omega}^{}\left( k_{S} \Delta S(x,t)-\dfrac{f(S^*,I^*)}{f(S(x,t),I^*)} 
k_{S} \Delta S(x,t) + k_E \Delta E(x,t)-\dfrac{E^*}{E(x,t)} k_E \Delta E(x,t)\right)dx\\
&+\displaystyle \int_{\Omega}^{} \left( k_{I} \Delta I(x,t)
-\dfrac{I^*}{I(x,t)} k_{I} \Delta I(x,t)\right) dx.
\end{split}
\end{equation*}
Therefore,
\begin{equation*}
\begin{split}
\dfrac{dW(t)}{dt}
&= \displaystyle \int_{\Omega}^{} \mu 
\left(1-\dfrac{f(S^*,I^*)}{f(S(x,t),I^*)}\right) \Big(S(x,t)-S^* \Big) dx
+ \beta f(S^*,I^*)\\
&\times \left[\displaystyle \int_{\Omega}^{} \displaystyle 
\int_0^{h}g(\tau) \left(-1-\dfrac{f(S(x,t),I(t-\tau)}{f(S(x,t),I^*)} 
- \dfrac{I(x,t)}{I^*}  +\dfrac{I(x,t)}{I^*}
\dfrac{f(S(x,t),I^*)}{f(S(x,t),I(x,t- \tau))}\right) dx\right.\\
&+ \displaystyle \int_{\Omega}^{} \displaystyle 
\int_0^{h}g(\tau) \left(4-\dfrac{f(S^*,I^*)}{f(S(x,t),I^*)} 
-\dfrac{E(x,t)I^*}{I(x,t) E^*} -\dfrac{I(x,t)}{I^*}
\dfrac{f(S(x,t),I^*)}{f(S(x,t),I(x,t- \tau))}\right)dx\\
&\left. -\displaystyle \int_{\Omega}^{} 
\dfrac{f(S(x,t),I(x,t- \tau))E^*}{f(S^*,I^*)E} dx\right]\\
&+ \displaystyle \int_{\Omega}^{}  \left( k_{S} \Delta S(x,t)
-\dfrac{f(S^*,I^*)}{f(S(x,t),I^*)} k_{S} \Delta S(x,t) 
+ k_E \Delta E(x,t)-\dfrac{E^*}{E(x,t)} k_E \Delta E(x,t)\right)dx\\
&+\displaystyle \int_{\Omega}^{} \left( k_{I} \Delta I(x,t)
-\dfrac{I^*}{I(x,t)} k_{I} \Delta I(x,t)\right) dx
\end{split}
\end{equation*}
and
\begin{equation*}
\begin{split}	
-&1-\dfrac{f(S(x,t),I(t-\tau)}{f(S(x,t),I^*)} - \dfrac{I(x,t)}{I^*}
+\dfrac{I(x,t)}{I^*}\dfrac{f(S(x,t),I^*)}{f(S(x,t),I(x,t-\tau))}\\
&=\left(\dfrac{I^*}{I(x,t)}-\dfrac{f(S(x,t),I(t-\tau))}{f(S(x,t),I^*)}\right)
\left(\dfrac{f(S(x,t),I^*)}{f(S(x,t),I(x,t-\tau))}-1 \right)\\
&\leq 0.
\end{split}
\end{equation*}
In view of
\begin{multline*}
\ln\left(\dfrac{f(S^*,I^*)}{f(S(x,t),I^*)}\right)
+ \ln\left(\dfrac{E(x,t)I^*}{I(x,t) E^*}\right) 
+ \ln\left(\dfrac{I(x,t)}{I^*}\dfrac{f(S(x,t),I^*)}{f(S(x,t),I(x,t- \tau))}\right)\\
+\ln\left(\dfrac{f(S(x,t),I(x,t- \tau))E^*}{f(S^*,I^*)E}\right)=0
\end{multline*}
we obtain that
\begin{equation*}
\begin{split}	
4&-\dfrac{f(S^*,I^*)}{f(S(x,t),I^*)} -\dfrac{E(x,t)I^*}{I(x,t) E^*} 
-\dfrac{I(x,t)}{I^*}\dfrac{f(S(x,t),I^*)}{f(S(x,t),I(x,t- \tau))}\\
&= G\left(\dfrac{f(S^*,I^*)}{f(S(x,t),I^*)}\right) 
+ G\left(\dfrac{E(x,t)I^*}{I(x,t) E^*}\right) + G\left(\dfrac{I(x,t)}{I^*}
\dfrac{f(S(x,t),I^*)}{f(S(x,t),I(x,t- \tau))}\right)\\
&\quad +G\left(\dfrac{f(S(x,t),I(x,t- \tau))E^*}{f(S^*,I^*)E}\right)\\
& \leq 0.
\end{split}
\end{equation*}
By Green's formula and from the Neumann boundary conditions, it follows that
\begin{multline*}
\int_{\Omega}^{} \left(k_S \Delta S(x,t) 
+ k_E  \Delta E(x,t)+ k_I  \Delta I(x,t)\right) dx \\
=  \displaystyle \int_{\partial\Omega}^{} 
\left(k_S \dfrac{\partial S(x,t)}{\partial \nu}
+k_E \dfrac{\partial E(x,t)}{\partial \nu} 
+ k_I \dfrac{\partial I(x,t)}{\partial \nu}\right) dx=0,
\end{multline*}
$$
\int_{\Omega}^{}\dfrac{f(S^*,I^*)}{f(S(x,t),I^*)}k_{S} \Delta S(x,t) dx
= k_{S} f(S^*,I^*)\displaystyle 
\int_{\Omega}^{}\dfrac{\partial f(S,I^*)}{\partial S}\dfrac{|\nabla S|^2}{I^2}dx,
$$
$$
\int_{\Omega}^{}\dfrac{I^*}{I(x,t)} k_{I} \Delta I(x,t) dx
= I^* k_{I} \displaystyle \int_{\Omega}^{}\dfrac{|\nabla S|^2}{I^2}dx
$$ 
and 
$$
\int_{\Omega}^{}\dfrac{E^*}{E(x,t)} k_E \Delta E(x,t) dx
= E^* k_E \displaystyle \int_{\Omega}^{}\dfrac{|\nabla E|^2}{E^2}dx.
$$
From hypothesis $(H_1)$, we have
$$  
\dfrac{\partial f(S,I^*)}{\partial S}>0.
$$
Hence, for any $t>0$, $R_0>1$ ensures 
$$
\dfrac{dW(t)}{dt} \leq 0 \ \text{ for all } S, E, I\geq 0
$$ 
and 
$$
\dfrac{dW(t)}{dt}= 0\ 
\text{ if and only if } S=S^*, E=E^* \text{ and } I=I^*.
$$ 
Clearly, the largest compact invariant set in 
$$
\left \{(S, I, R): \dfrac{dW(t)}{dt}= 0\right \}
$$ 
is the singleton $\{E^*\}$. By applying LaSalle's invariance 
principle \cite[Theorem 4.3.4]{Daniel1981}, we conclude that 
the endemic equilibrium point of system \eqref{MM} 
is globally asymptotically stable. The proof is complete.
\end{proof}


\section{Numerical simulations}
\label{sec:4}

In this section, we do numerical simulations 
in order to illustrate our analytical results. Let 
\begin{equation}
\label{example}
\left\{
\begin{array}{lll}
\dfrac{\partial S(x,t)}{\partial t} - k_S \Delta S(x,t)
= b- \mu S(x,t)- \beta \displaystyle \int_0^{h}g(\tau)
S(x,t)I(x,t-\tau)d\tau, \quad x\in\Omega,\\[2ex]
\dfrac{\partial E(x,t)}{\partial t} - k_E \Delta E(x,t)
= \beta \displaystyle \int_0^{h}g(\tau)S(x,t)I(x,t-\tau)d\tau 
-(\mu+\alpha) E(x,t), \quad x\in\Omega,\\[2ex]
\dfrac{\partial I(x,t)}{\partial t} - k_I \Delta I(x,t)
= \alpha E(x,t) -(\mu+c+\gamma)I(x,t), \quad x\in\Omega,\\[2ex]
\dfrac{\partial R(x,t)}{\partial t}- k_R \Delta R(x,t)
=\gamma I(x,t)- \mu R(x,t), \quad x\in\Omega,\\[2ex]
\displaystyle \frac{\partial S(x,t)}{\partial \nu}
=\frac{\partial E(x,t)}{\partial \nu}
=\frac{\partial I(x,t)}{\partial \nu}=0, 
\frac{\partial R(x,t)}{\partial \nu}=0, 
\quad x\in\partial\Omega,
\end{array}
\right.
\end{equation}
$t>0$. Here, function $g$ takes the following form:
$$ 
g(\tau) = \dfrac{1}{h}, \quad h>0.
$$
The basic reproduction number $R_{0}$ is given by 
$$ 
R_{0} = \dfrac{ \beta b \alpha}{ \mu(\mu+\alpha)( \mu + c + \gamma)}.
$$
We consider the following initial conditions:
$$ 
S(0,x)= 50, \  E(0, x)= 2, \  I(\theta, x)= 8 
\mbox{ and } \ R(0, x)=5, 
$$
where $\theta \in [-\tau, 0]$ and $x\in[0,50]$.
We first focus on a one-dimensional domain, which can be taken, 
without loss of generality, as being $[0, 50]$.
	
In order to solve numerically the considered system, 
we have used the method of centered finite differences 
to approximate the Laplacian, 	
$$ 
\dfrac{\partial^2 u}{\partial x^2}
=\dfrac{u^n_{i+1,j}- 2u^n_{i,j}+u^n_{i-1,j}}{\Delta x^2}
$$
with $\Delta x$ being the space discretization step and $u$ a given function.
We choose this method because it gives a precision of order 2 in space. 
Temporal discretization is performed using an explicit scheme.
The approximation of the term 
$$
\int_0^{h}S(x,t)I(x,t-\tau) d \tau
$$ 
is performed using the method of rectangles. Homogeneous Neumann boundary conditions (zero flux) 
are also approached by the method of centered finite differences in order to not lose 
the order of convergence of our scheme. The graphical visualization of numerical solutions, 
in space and time, was carried out using  $\mbox{{\sc Matlab}}^{\tiny \circledR}$. 
The used values of the parameters of the model are $b=5$, $\mu=0.1$, $\gamma=0.02$, 
$k_{S}=k_E=k_{I}=k_R=0.001$, $\beta=0.004$, $c=0.001$, $h=0.0001$, and $\alpha=0.1$;
and we take  $x \in [0,50]$ and $t \in [0, 300]$. 
The obtained results are shown in Figures~\ref{figure1}--\ref{figure4}.

\begin{figure}[ht!]
\begin{minipage}[b]{0.5\linewidth}
\begin{center}
\includegraphics[scale=0.55]{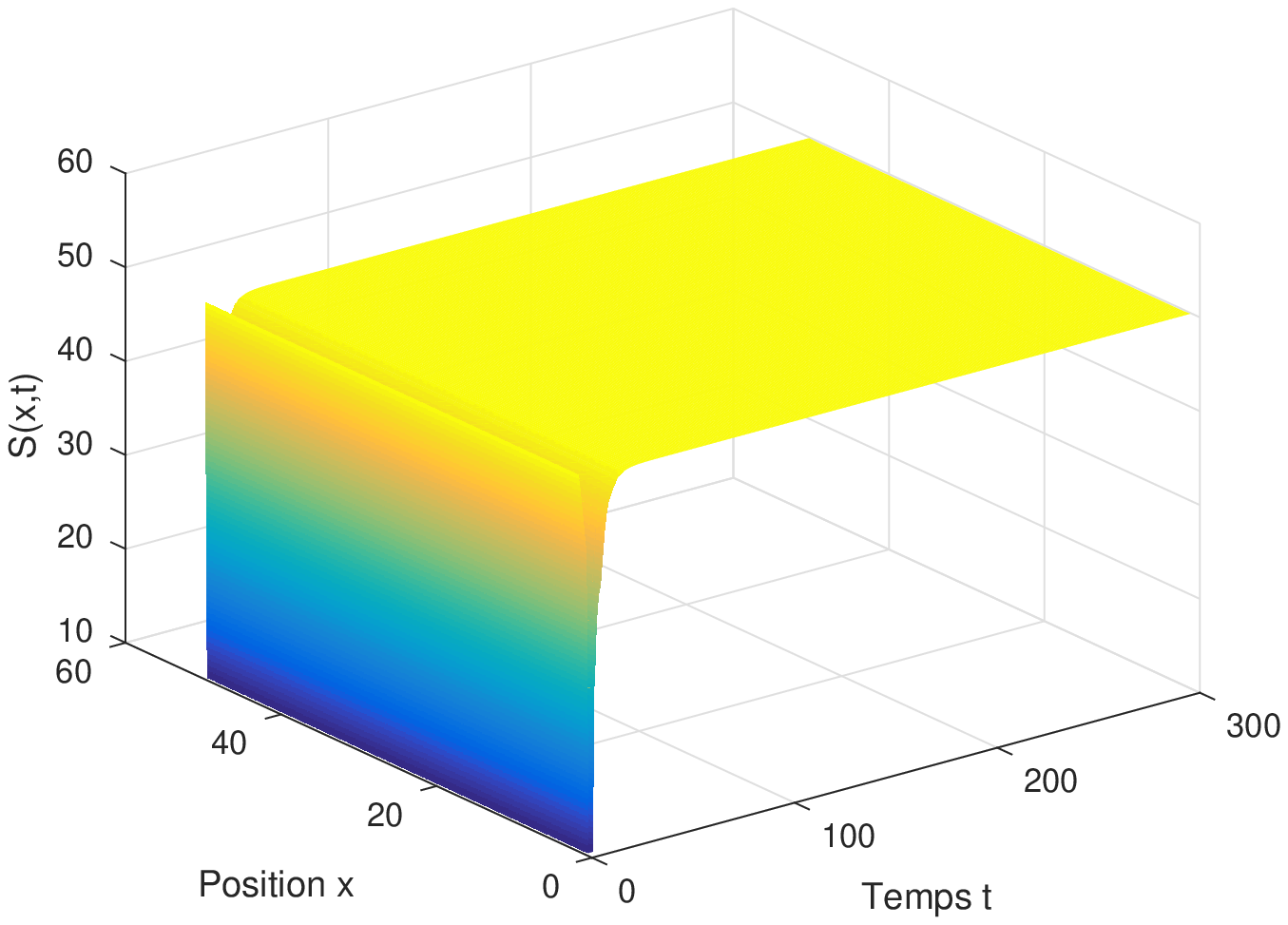}
\end{center}
\centering 
\end{minipage}\hfill
\begin{minipage}[b]{0.5\linewidth}
\begin{center}
\includegraphics[scale=0.55]{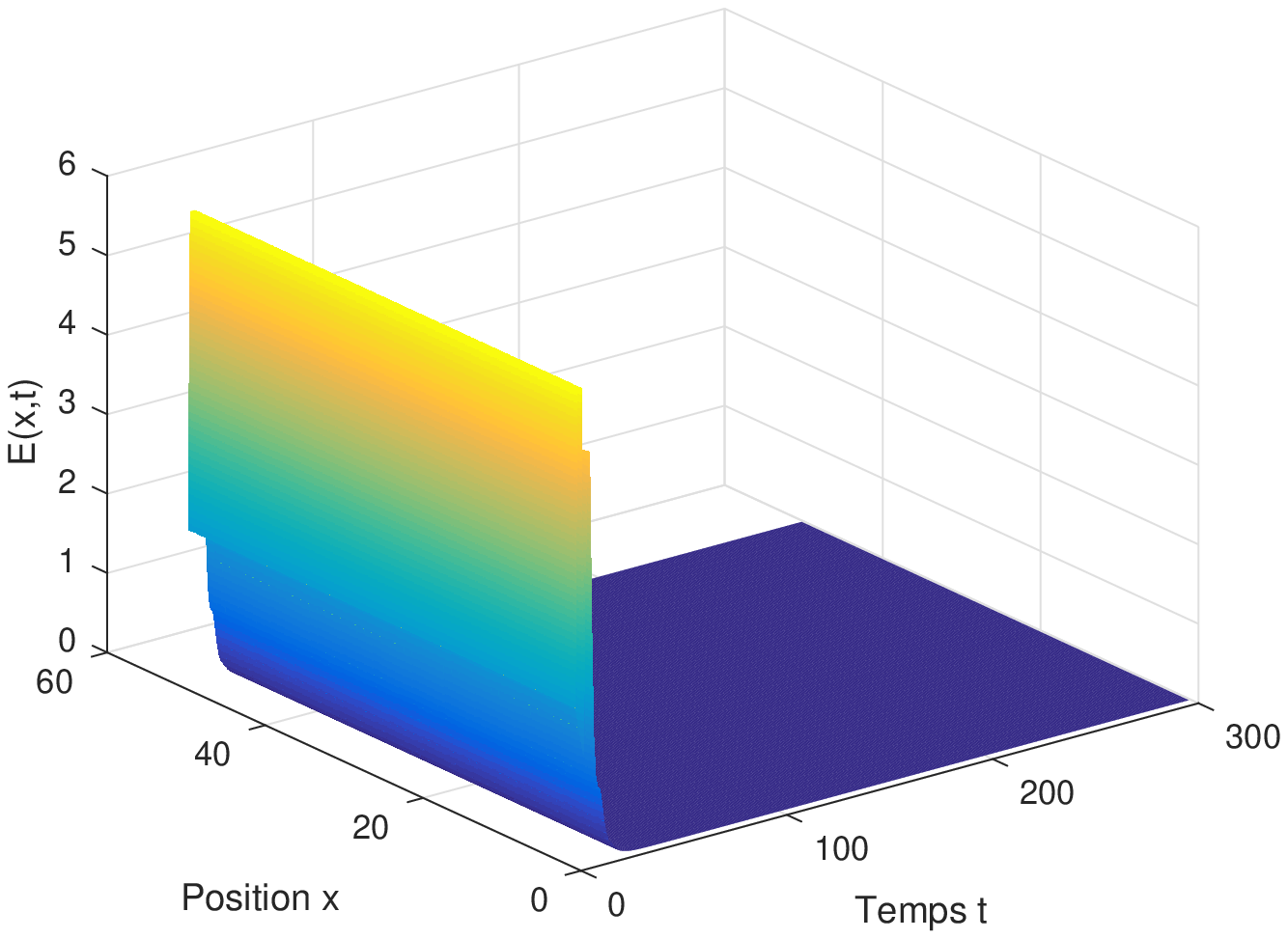}
\end{center}
\centering 
\end{minipage}
\captionof{figure}{ \label{figure1}  
Evolution of $S$ and $E$ of model  \eqref{example} for $x$ fixed 
and with the parameters described in Section~\ref{sec:4}. 
In this case, $R_{0}=0.8264< 1$.}	
\end{figure}
\begin{figure}[ht!]
\begin{minipage}[b]{0.5\linewidth}
\begin{center}
\includegraphics[scale=0.42]{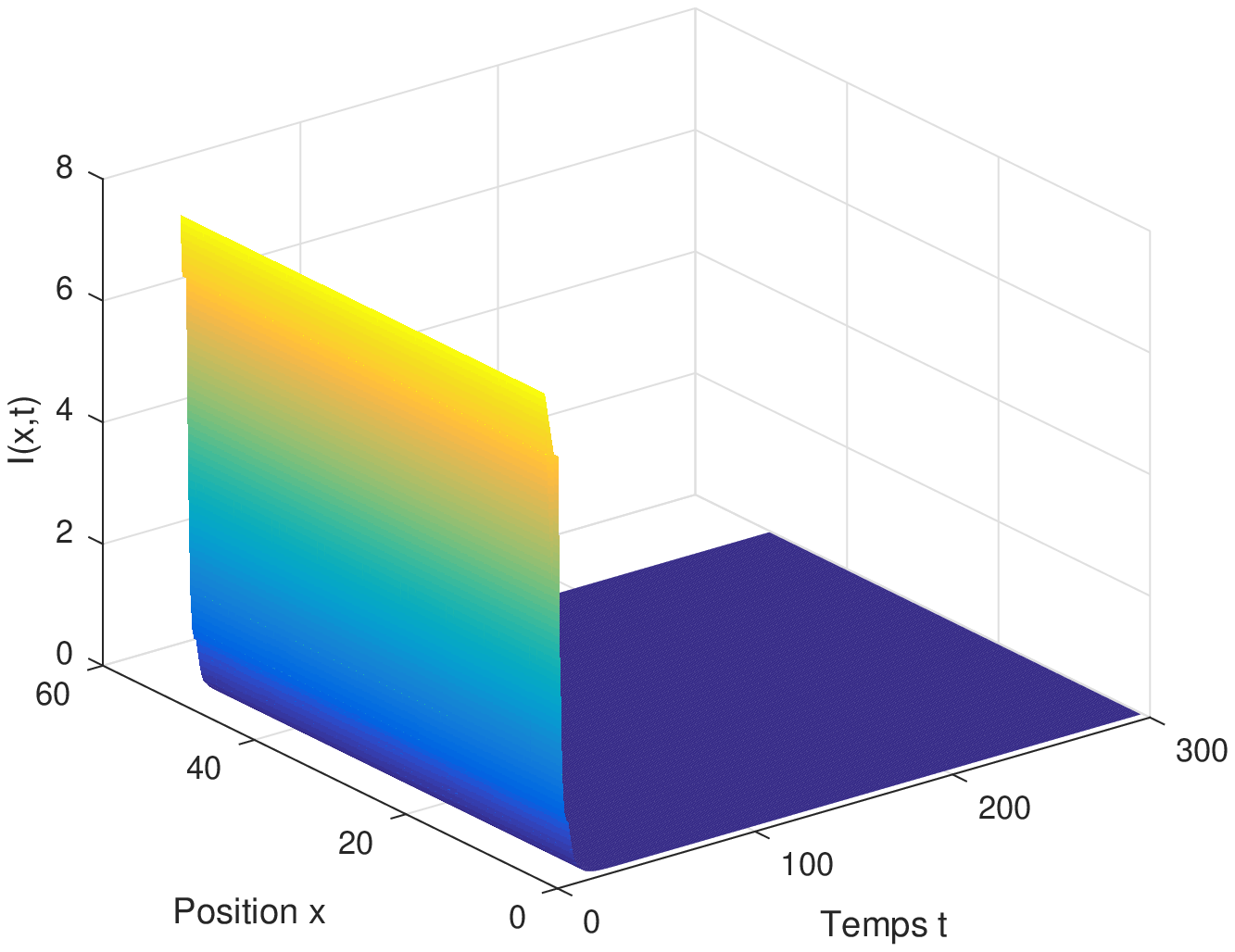}
\end{center}
\end{minipage}
\begin{minipage}[b]{0.5\linewidth}
\begin{center}
\includegraphics[scale=0.42]{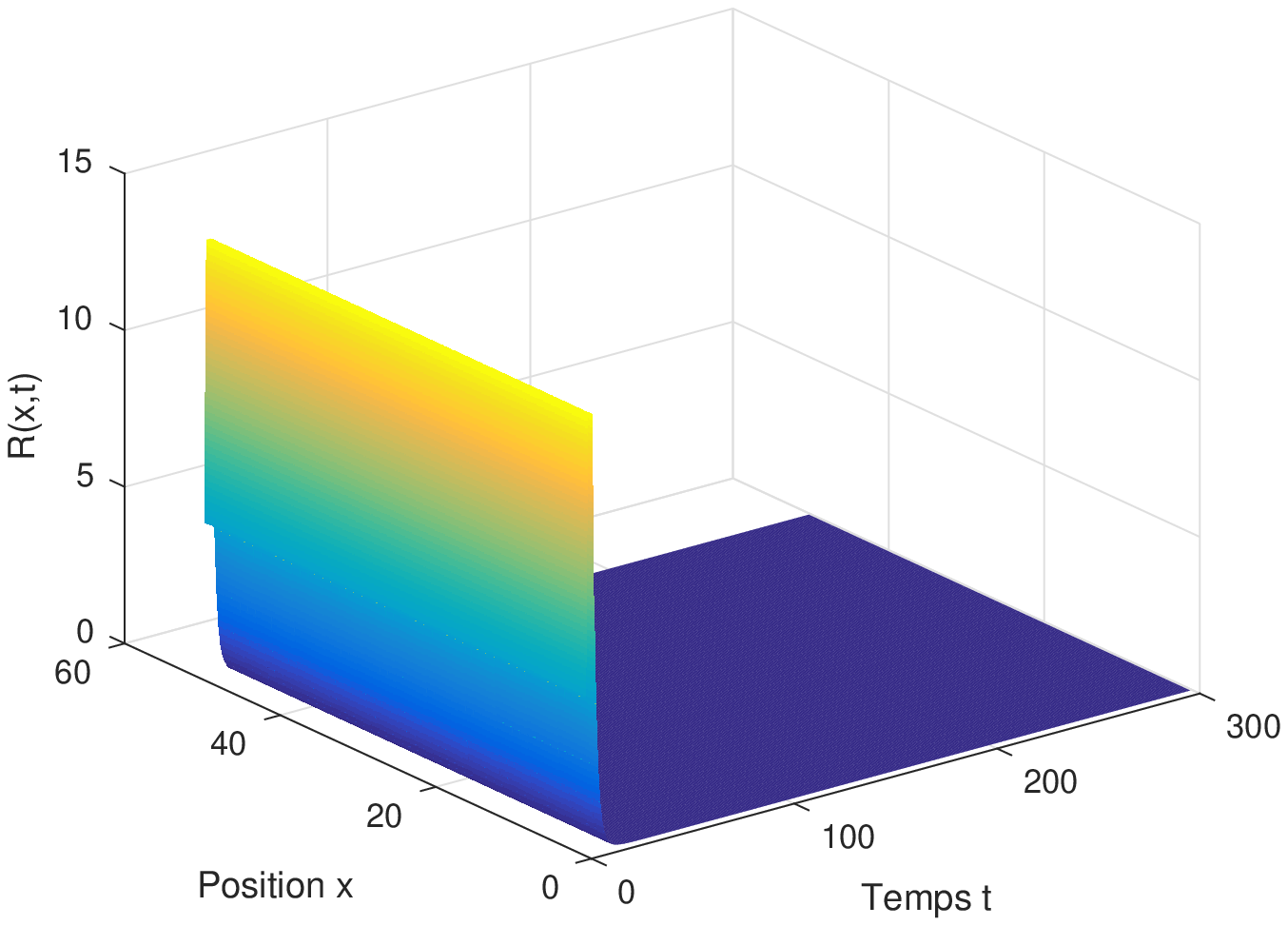}
\end{center}
\end{minipage}
\captionof{figure}{ \label{figure2} 
Evolution of $I$ and $R$ of model \eqref{example} for $x$ fixed  
and with the parameters described in Section~\ref{sec:4}. 
In this case, $R_{0}=0.8264< 1$.}	 	
\end{figure}
\begin{figure}[ht!]
\begin{minipage}[b]{0.5\linewidth}
\begin{center}
\includegraphics[scale=0.42]{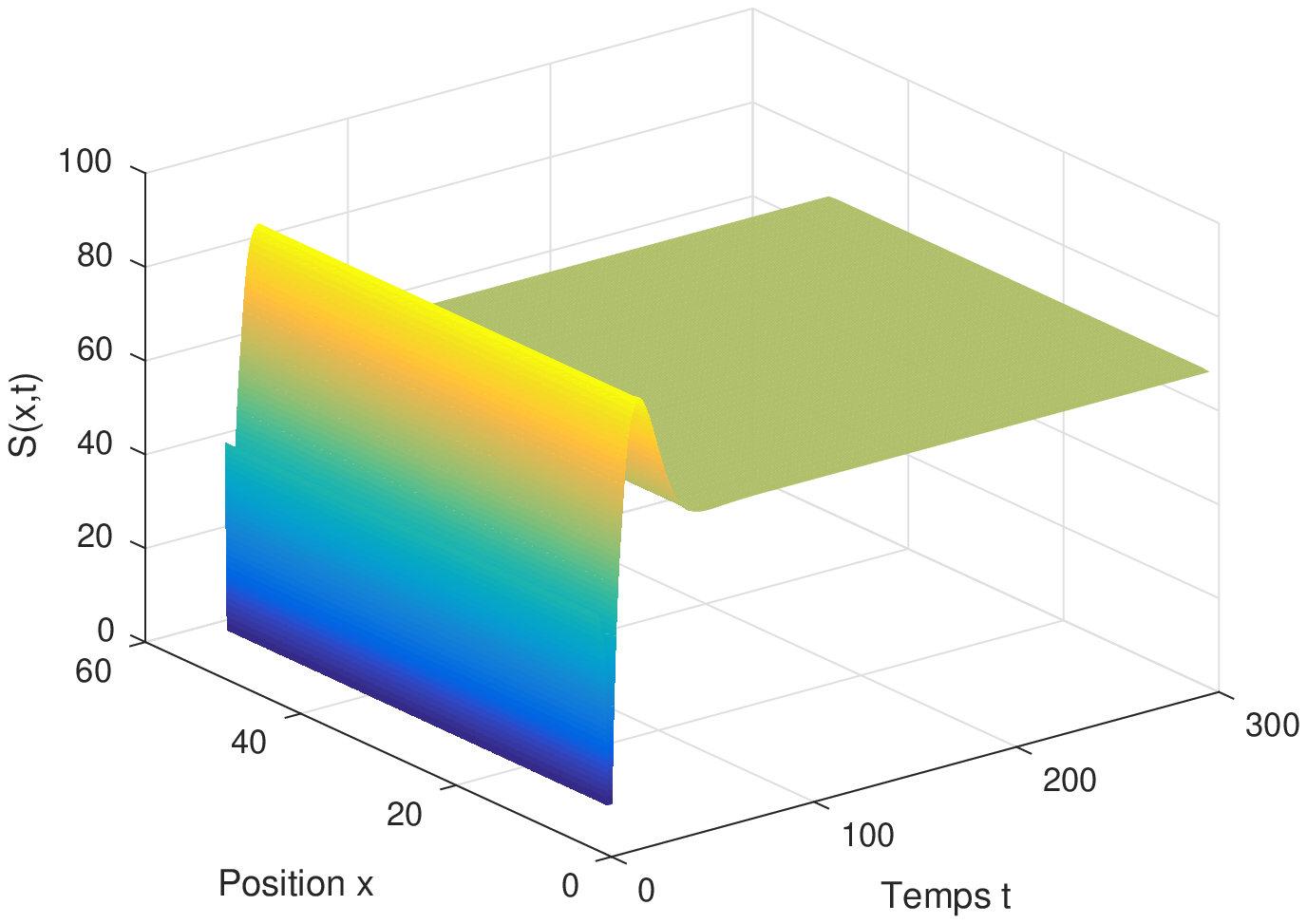}
\end{center}
\end{minipage}
\begin{minipage}[b]{0.5\linewidth}
\begin{center}
\includegraphics[scale=0.42]{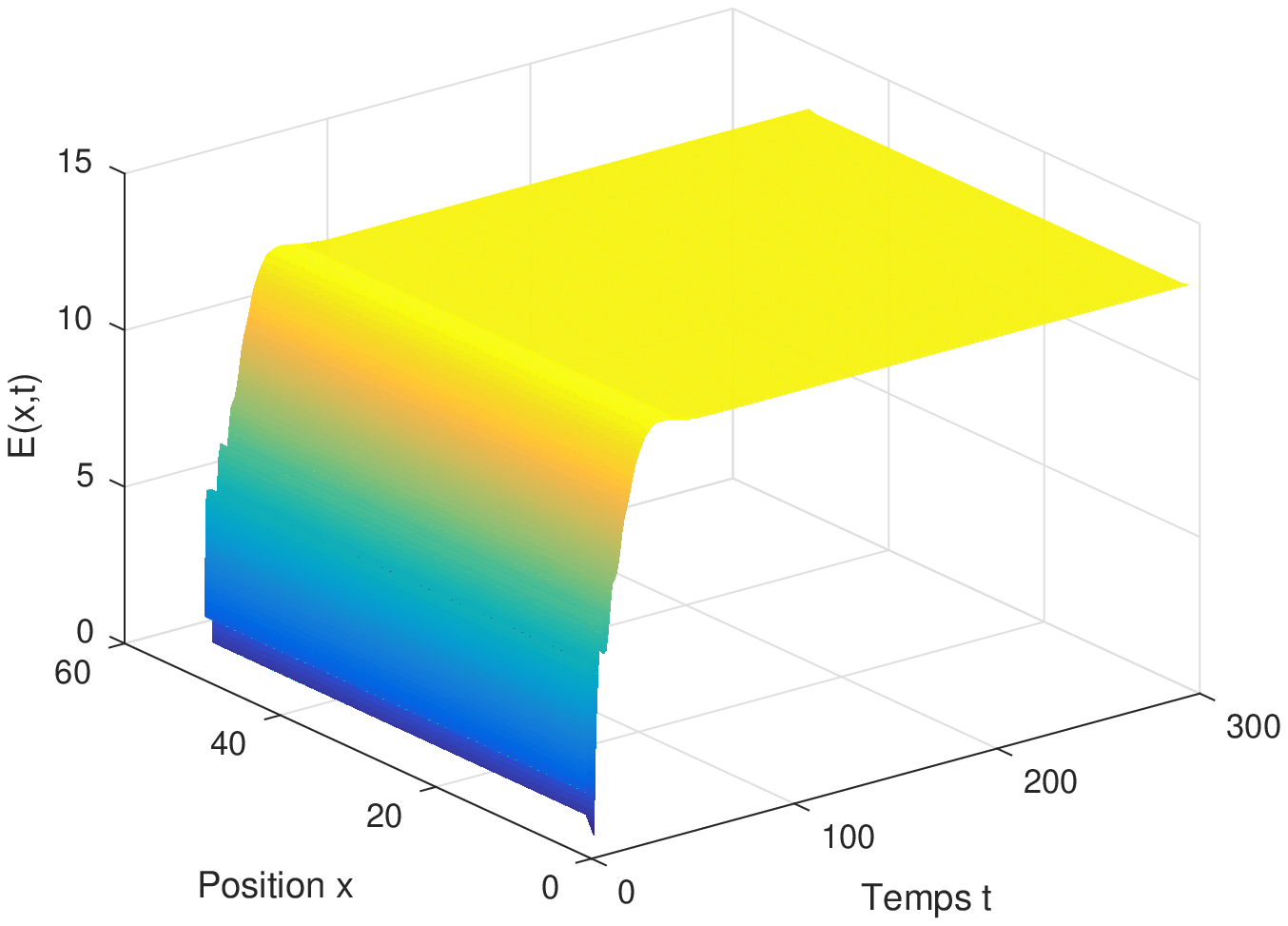}
\end{center}
\centering 
\end{minipage}\hfill
\captionof{figure}{ \label{figure3}  
Evolution of $S$ and $E$ of model \eqref{example}  
with the parameters described in Section~\ref{sec:4}, 
excepting $\mu =0.03$ and $\alpha=0.2$. In this case, $R_{0}=11.3669> 1$.}	
\end{figure}
\begin{figure}[ht!]
\begin{minipage}[b]{0.5\linewidth}
\begin{center}
\includegraphics[scale=0.42]{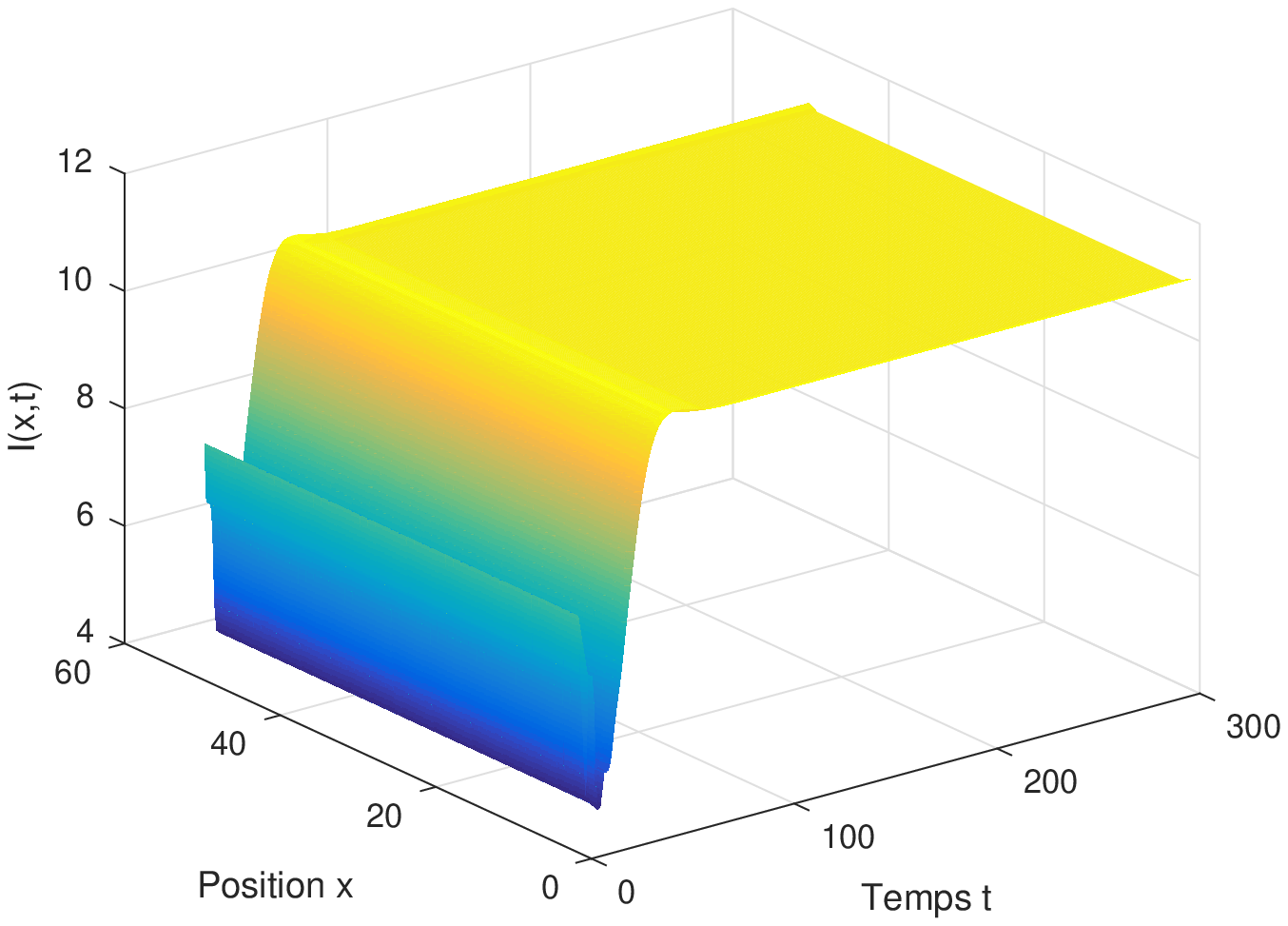}
\end{center}
\centering 
\end{minipage}
\begin{minipage}[b]{0.5\linewidth}
\begin{center}
\includegraphics[scale=0.42]{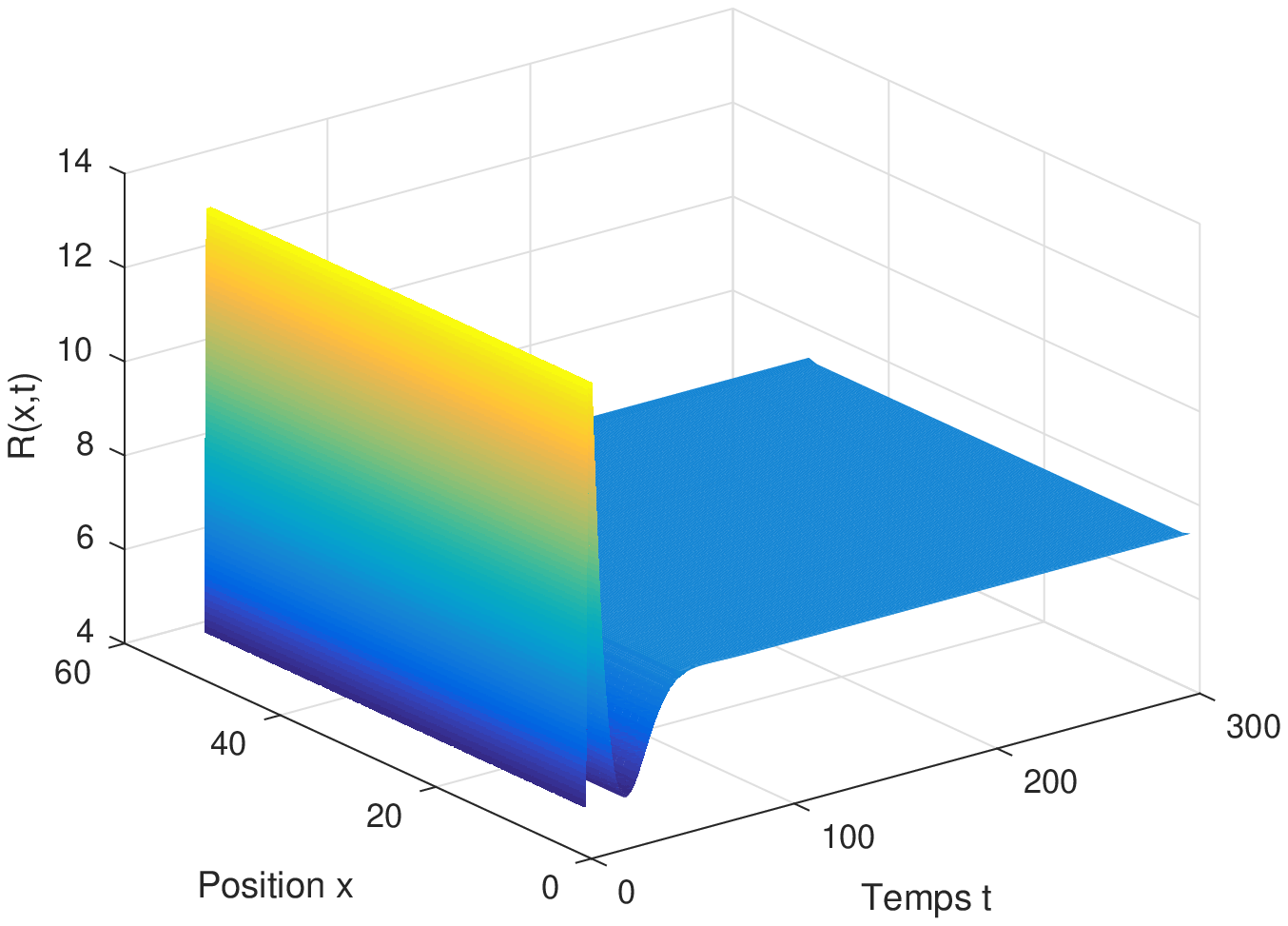}
\end{center}
\end{minipage}
\captionof{figure}{ \label{figure4} 
Evolution of $I$ and $R$ of model \eqref{example}  
with the parameters described in Section~\ref{sec:4},
excepting $\mu =0.03$ and $\alpha=0.2$. In this case, $R_{0}=11.3669> 1$.}	
\end{figure}

From Figures~\ref{figure1} and \ref{figure2}, we note that the solution 
$(S(t), E(t), I(t),R(t))$ of system \eqref{example} converges to the 
free disease equilibrium $E_0=\left(50,0, 0,0\right)$. In other words, 
$E_0$ is globally asymptotically stable. From a biological point of view, 
if $R_{0} \leq 1$, then the infection can be eradicated from the population. 
In addition, from Figures~\ref{figure3} and \ref{figure4},  we can draw 
the following conclusion: for $ R_{0}> 1 $, the solution   
$(S(t), E(t), I(t), R(t))$ of model \eqref{example} converges to the endemic 
equilibrium $E^*$. Thus, the unique endemic equilibrium is globally asymptotically stable, 
which biologically means that the infection persists but is controlled.

It should be noted that our numerical simulations can be 
performed, without any difficulty, for dimension two in space.
	

\section{Concluding remarks}
\label{sec:5}

We have studied the qualitative behavior of solutions of a 
reaction-diffusion system with distributed delay and a general 
nonlinear incidence function. We have shown that the model exhibits 
two equilibria: a disease-free equilibrium $E_0$ and an endemic equilibrium $E^*$. 
Under some assumptions on the incidence function, we have shown that 
the global dynamics of the model is completely determined by  
the basic reproduction number $R_0$. More precisely, we have proved 
that $R_0$ serves as a threshold parameter for the persistence and extinction of the disease.  
Since the coefficients of the system \eqref{Model} are all constants, 
we took the advantages of the method of Lyapunov functions to obtain 
the global dynamics of the considered model, showing that the disease 
free equilibrium state is globally asymptotically stable for $R_0\leq 1$. 
When $ R_0 > 1$, then we proved that there is a unique disease 
endemic equilibrium, which is globally asymptotically stable. 
Epidemiologically, this means that the disease will die out or will persist 
in the population depending on the values of the parameters of the model.
	

\section*{Acknowledgment}

Torres was supported by FCT within project UIDB/04106/2020 (CIDMA).




\begin{thebibliography}{xx}

\expandafter\ifx\csname url\endcsname\relax
\def\url#1{\texttt{#1}}\fi
\expandafter\ifx\csname urlprefix\endcsname\relax\def\urlprefix{URL }\fi
\expandafter\ifx\csname href\endcsname\relax
\def\href#1#2{#2} \def\path#1{#1}\fi
	
\bibitem{SAMSUZZOHA20103461}
M.~Samsuzzoha, M.~Singh, D.~Lucy, 
Numerical study of an influenza epidemic model with diffusion, 
Applied Mathematics and Computation 217~(7) (2010) 3461--3479.
	
\bibitem{SAMSUZZOHA20115507}
M.~Samsuzzoha, M.~Singh, D.~Lucy, 
Numerical study of a diffusive epidemic model of influenza 
with variable transmission coefficient, 
Applied Mathematical Modelling 35~(12) (2011) 5507--5523.
	
\bibitem{Bai2018}
Z.~Bai, R.~Peng, X.~Q. Zhao, 
A reaction-diffusion malaria model with seasonality and incubation period, 
Journal of Mathematical Biology 77~(1) (2018) 201--228.
	
\bibitem{Banerjee2016}
M.~Banerjee, V.~Volpert, 
Spatio-temporal pattern formation in Rosenzweig--MacArthur model: 
Effect of nonlocal interactions, 
Ecological Complexity 30 (2017) 2--10, 
Dynamical Systems In Biomathematics.
\newblock \href {https://doi.org/10.1016/j.ecocom.2016.12.002}
{\path{doi:10.1016/j.ecocom.2016.12.002}}.
	
\bibitem{Hwang2013}
T.~W. Hwang, F.~B. Wang, 
Dynamics of a dengue fever transmission model with crowding effect 
in human population and spatial variation, 
Discrete \& Continuous Dynamical Systems -- B 18 (2013) 147--161.
	
\bibitem{Lou2011}
Y.~Lou, X.~Q. Zhao, 
A reaction-diffusion malaria model with incubation period in the vector population, 
Journal of Mathematical Biology 62~(4) (2011) 543--568.
	
\bibitem{Wang2018}
S.~Wang, 
Threshold dynamics of an {SIR} epidemic model 
with nonlinear incidence rate and non-local delay effect,
Wuhan University Journal of Natural Sciences 23~(6) (2018) 503--513.
	
\bibitem{Xu2017}
Z.-t. Xu, D.-x. Chen, 
An {SIS} epidemic model with diffusion, 
Appl. Math. J. Chinese Univ. Ser. B 32~(2) (2017) 127--146.
\newblock \href {https://doi.org/10.1007/s11766-017-3460-1}
{\path{doi:10.1007/s11766-017-3460-1}}.
	
\bibitem{Nauman2019}
N.~Ahmed, Z.~Wei, D.~Baleanu, M.~Rafiq, M.~A. Rehman, 
Spatio-temporal numerical modeling of reaction-diffusion measles epidemic system, 
Chaos: An Interdisciplinary Journal of Nonlinear Science 29~(10) (2019) 103101.

\bibitem{Jinliang2017}
T.~Kuniya, J.~Wang, 
Lyapunov functions and global stability for a spatially diffusive {SIR} epidemic model, 
Applicable Analysis 96~(11) (2017) 1935--1960.
	
\bibitem{KIM20131992}
K.~I. Kim, Z.~Lin, Q.~Zhang, 
An {SIR} epidemic model with free boundary,
Nonlinear Analysis: Real World Applications 14~(5) (2013) 1992--2001.
	
\bibitem{CAPASSO197843}
V.~Capasso, G.~Serio, 
A generalization of the {Kermack-McKendrick} deterministic epidemic model, 
Mathematical Biosciences 42~(1) (1978) 43--61.
	
\bibitem{Abdesslem}
A.~Elazzouzi, A.~Lamrani~Alaoui, M.~Tilioua, D.~F.~M. Torres, 
Analysis of a {SIRI} epidemic model with distributed delay and relapse, 
Statistics, Optimization and Information Computing 7 (2019) 545--557.
{\tt arXiv:1812.09626}
	
\bibitem{enatsu2012lyapunov}
Y.~Enatsu, 
Lyapunov functional techniques on the global stability of equilibria 
of {SIS} epidemic models with delays, 
Kyoto Univ. Res. Inf. Repository 1792 (2012) 118--130.
	
\bibitem{MR2762609}
A.~Kaddar, 
Stability analysis in a delayed {SIR} epidemic model 
with a saturated incidence rate, 
Nonlinear Anal. Model. Control 15~(3) (2010) 299--306.
	
\bibitem{MR2130156}
A.~Korobeinikov, P.~K. Maini, 
A {L}yapunov function and global properties for {SIR} and {SEIR} 
epidemiological models with nonlinear incidence, 
Math. Biosci. Eng. 1~(1) (2004) 57--60.
	
\bibitem{MR2879132}
A.~Lahrouz, L.~Omari, D.~Kiouach, A.~Belma\^ati, 
Complete global stability for an {SIRS} epidemic model 
with generalized non-linear incidence and vaccination, 
Appl. Math. Comput. 218~(11) (2012) 6519--6525.
	
\bibitem{WANG20102390}
J.~J. Wang, J.~Z. Zhang, Z.~Jin, 
Analysis of an {SIR} model with bilinear incidence rate, 
Nonlinear Analysis: Real World Applications 11~(4) (2010) 2390--2402.
	
\bibitem{gumel2003qualitative}
A.~B. Gumel, S.~M. Moghadas, 
A qualitative study of a vaccination model with non-linear incidence, 
Applied Mathematics and Computation 143 (2003) 409--419.
	
\bibitem{LI2014100}
J.~Li, G.~Q. Sun, Z.~Jin, 
Pattern formation of an epidemic model with time delay, 
Physica A: Statistical Mechanics and its Applications 403 (2014) 100--109.
	
\bibitem{Ruan2003135}
S.~Ruan, W.~Wang, 
Dynamical behavior of an epidemic model with a nonlinear incidence rate, 
Journal of Differential Equations 188 (2003) 135--163.

\bibitem{Elazzouzi}
A.~Elazzouzi, A.~Lamrani~Alaoui, M.~Tilioua, A.~Tridane, 
Global stability analysis for a generalized delayed {SIR} model 
with vaccination and treatment, 
Advances in Difference Equations 532~(2019) (2019).
	
\bibitem{Jing2011}
J.~Yang, S.~Liang, Y.~Zhang, 
Travelling waves of a delayed {SIR} epidemic model 
with nonlinear incidence rate and spatial diffusion, 
PLOS ONE 6 (2011) 1--14.
	
\bibitem{Xia2018}
W.~Xia, S.~Kundu, S.~Maitra, 
Dynamics of a delayed {SEIQ} epidemic model,
Advances in Difference Equations 2018 (2018) 336.
	
\bibitem{CONNELLMCCLUSKEY201564}
C.~C. McCluskey, Y.~Yang, 
Global stability of a diffusive virus dynamics model
with general incidence function and time delay, 
Nonlinear Anal. Real World Appl. 25 (2015) 64--78.
\newblock \href {https://doi.org/10.1016/j.nonrwa.2015.03.002}
{\path{doi:10.1016/j.nonrwa.2015.03.002}}.
	
\bibitem{junjie2020}
H.~Yang, J.~Wei, 
Dynamics of spatially heterogeneous viral model with time delay, 
Communications on Pure \& Applied Analysis 19 (2020) 85--102.
	
\bibitem{McCluskey2010837}
C.~C. McCluskey, 
Global stability of an {SIR} epidemic model with delay 
and general nonlinear incidence, 
Math. Biosci. Eng. 7~(4) (2010) 837--850.
\newblock \href {https://doi.org/10.3934/mbe.2010.7.837}
{\path{doi:10.3934/mbe.2010.7.837}}.
	
\bibitem{hale}
J.~K. Hale, 
Theory of functional differential equations, 
Part of the Applied Mathematical Sciences book series 
(AMS, volume 3), Springer-Verlag (1977).
	
\bibitem{hale2}
J.~K. Hale, 
Ordinary Differential Equations, 
Krieger publishing company Malabar, Florida, 1980.

\bibitem{Daniel1981}
D.~Henry, 
Geometric Theory of Semilinear Parabolic Equations, 
Springer, New York, 1981.
	
\end{thebibliography}
\end{document}